\documentclass[conference]{IEEEtran}
\newcommand{\nop}[1]{}

\nop{
\newcommand{\nop}[1]{}
\topmargin       -10.0mm
\oddsidemargin   -5.0mm
\evensidemargin  -5.0mm
\textheight     235mm
\textwidth      170.0mm
}

\usepackage{amssymb, amsmath}
\usepackage{amsthm}
\usepackage{graphicx}  
\usepackage{bbm}

\newtheorem{definition}{Definition}

\newtheorem{theorem}{Theorem}
\newtheorem{corollary}{Corollary}
\newtheorem{lemma}{Lemma}

\begin{document}

\title{A Basic Result on the Superposition of Arrival Processes in Deterministic Networks}   
\author{Yuming Jiang \\NTNU, Norwegian University of Science and Technology,~Norway}        
\date{\today}          
\maketitle

\begin{abstract}
Time-Sensitive Networking (TSN) and Deterministic Networking (DetNet) are emerging standards to enable deterministic, delay-critical communication in such networks. This naturally (re-)calls attention to the network calculus theory (NC), since a rich set of results for delay guarantee analysis have already been developed there. One could anticipate an immediate adoption of those existing network calculus results to TSN and DetNet. However, the fundamental difference between the traffic specification adopted in TSN and DetNet and those traffic models in NC makes this difficult, let alone that there is a long-standing open challenge in NC. To address them, this paper considers an arrival time function based max-plus NC traffic model. 
In particular, a relationship between the TSN / DetNet traffic specification and the NC traffic model is proved. In addition, the superposition property of the arrival time function based NC traffic model is found and proved. 
Appealingly, the proved superposition property shows a clear analogy with that of a well-known counterpart traffic model in NC. These results help make an important step towards the development of a system theory for delay guarantee analysis of TSN / DetNet networks.

\end{abstract}

\section{Introduction} \label{sec-1}

Time-sensitive applications can be broadly found in industrial process control,  machine control and live streaming of audio and video. To support such applications and to enable deterministic delay-critical communication, Time-Sensitive Networking (TSN) and Deterministic Networking (DetNet) are emerging standards respectively introduced by the IEEE TSN Task Group \cite{TSN-802.1Qcc} for Layer 2 Ethernet switches, and by the IETF DetNet Working Group \cite{IETF-DetNet} for more general network settings. In both TSN and DetNet, the traffic specification (TSpec) uses two parameters to model a flow or an arrival process: {\em a time interval} and {\em the maximum number of packets in the interval} \cite{TSN-802.1Qcc}\cite{IETF-DetNet}. 

Since a rich set of results for delay guarantee analysis have already been developed in the network calculus theory (NC), e.g. \cite{Cruz91a}\cite{Cruz91b}\cite{Chang00}\cite{NetCal}\cite{Sigcomm06}\cite{SNC}\cite{Jiang-GL09}\cite{Jiang-mascots09}\cite{Jiang-valuetools11}\cite{Liebeherr17}, one could anticipate an immediate adoption of those existing NC results to TSN and DetNet. However, the traffic specification adopted in TSN and DetNet is fundamentally different from those traffic models in NC, which makes the adoption difficult, let alone that there is a long-standing open challenge in the related part of NC. To address this difficulty and the open challenge forms the motivation and objective of the present paper. 

Specifically, in this paper, a packet arrival time function based traffic model related to the max-plus branch of NC \cite{Chang00}\cite{Jiang-GL09}\cite{Jiang-mascots09}\cite{Jiang-valuetools11}\cite{Liebeherr17} is introduced. We prove that there is a mapping between the TSN / DetNet traffic specification and the max-plus traffic model, which establishes an important link for making use of NC results to TSN / DetNet analysis. 

However, for this max-plus traffic model, there is a long-standing problem, which is its superposition property, i.e., the aggregate arrival process resulted from the aggregation of multiple arrival processes can be characterized using the same model as for the individual arrival processes. Specifically, the superposition property has surprisingly not been found or proved directly on the model itself for long time due to an inherent challenge \cite{Chang00}\cite{Jiang-GL09}\cite{Jiang-mascots09}\cite{Jiang-valuetools11}\cite{Liebeherr17}. Since the superposition property is one of the most basic properties needed for network performance analysis \cite{Sigcomm06}  \cite{SNC}, this calls for an urgent need of investigation. 

The inherent challenge \cite{Jiang-mascots09}\cite{Liebeherr17} is due to the complex formulation of the arrival time function of the aggregate process in terms of the arrival time functions of the individual arrival processes, making it difficult (if not impossible) to characterize the aggregate process directly on this aggregate arrival time function. 
To bypass this challenge, an indirect approach has been considered in the literature \cite{Chang00}\cite{Jiang-mascots09}.\nop{, which is to first transform the time-domain arrival characterization from the arrival time function to a model in another domain, then apply the superposition property to obtain a characterization for the aggregate process in that domain, and finally transform the obtained characterization back to the arrival time function characterization.} However, {\em this indirect approach requires packet length information \cite{Chang00}\cite{Jiang-mascots09}\cite{Liebeherr17}, which is not available or needed} in the arrival time function description of the arrival process as is the case in TSN and DetNet. 

In this paper, a novel approach is used which works directly on the arrival time functions, fundamentally different from the indirect approach. Based on this direct approach, the superposition property of the arrival time function based max-plus traffic model is found and proved. Appealingly, the proved superposition property has a clear analogy with the aggregation property of the well-known $(\sigma, \rho)$ traffic model  \cite{Cruz91a} in the min-plus branch of NC  \cite{Cruz91a}\cite{NetCal}, and is (much) better than that from the indirect approach. The superposition property and its proof using the direct approach form another contribution that is crucial to both NC and the future use of NC results to TSN and Det Net. 

The rest is organized as follows. In Sec. \ref{sec-2}, the max-plus traffic model is introduced, together with the proof of the mapping between it and the TSN / DetNet TSpec. In Sec. \ref{sec-3}, the inherent challenge is first discussed, followed by the superposition property with detailed proof. In Sec. \ref{sec-4}, a comparison study of  results using the indirect approach and the direct approach is provided. This comparison implies the importance of the superposition property proved in this paper. Finally, concluding remarks are made in Sec. \ref{sec-5}.

\section{The Max-Plus Traffic Model and The Mapping}\label{sec-2}

\subsection{Notation}
An arrival process is characterized by the arrival time function $\bar{A}(n)$, for $n=1, 2, \dots$, where $\bar{A}(n)$ denotes the arrival time of packet $n$. For notational convenience, we define $\bar{A}(0) = 0$. In addition, we define $\bar{A}(m, n) = \bar{A}(n) - \bar{A}(m)$ to be the inter-arrival time between the arrivals of packet $m$ and packet $n$, for $n \ge m \ge 1$. For instance, $\bar{A}(n, n + 1)$ is the inter-arrival time between packets $n$ and $n + 1$ for $n \ge 1$. 

As an analogy, we also characterize the arrival process using another function $A(t)$, $t \ge 0$, which counts the cumulative amount of traffic (in bits) carried by the arrival process up to time $t$. Similarly, we define $A(s, t) \equiv A(t) - A(s)$ as the cumulative amount of traffic carried by the arrival process from time $s$ to $t$, and for notational convenience, we let ${A}(0) = 0$. 

When studying the superposition of $I (\ge 2)$ multiple arrival processes, we use $\bar{A}_i(n)$, $(i=1, \dots, I)$, to denote the arrival time function of each individual arrival process, and $\bar{A}(n)$ that of the aggregate process. In addition,  we use ${A}_i(t)$, $(i=1, \dots, I)$, to denote the cumulative traffic amount time function of each individual arrival process, and $A(t)$ that of the aggregate process. 

\subsection{The TSN / DetNet Traffic Specification}

The TSN / DetNet traffic specification is defined as \cite{TSN-802.1Qcc}\cite{IETF-DetNet}:
\begin{definition}\label{def-tspec}
An arrival process is said to conform to the TSN / DetNet traffic specification with interval parameter $\tau (>0)$ and maximum packet number parameter $K (\ge 1)$, if during a specified duration of length $\tau$, the number of packets generated by this arrival process is limited by $K$. 
\end{definition}

For Definition~\ref{def-tspec}\nop{ this TSN / DetNet traffic specification}, we have the following remarks. First, this specification aims to characterize flows at the packet level. We believe, there is an underlying reason for this. In particular, the delay of a packet at a network node is comprised of two types of delays, namely processing related delays, and transmission related delays. \nop{The former includes the processing delay and the queueing delay due to waiting for processing. The latter includes the transmission delay and the queueing delay waiting for transmission.}Typically, delays in the first category are affected only at the packet level, little by the packet length, unlike the delays in the second category. With the link speed enters Gbps range, the nodal packet delay becomes more and more dominated by the first category, for which packet level characterization is crucial. 

Second, in \cite{TSN-802.1Qcc}\cite{IETF-DetNet}, there is a maximum packet length parameter that could also be included in the TSpec. However, by convention, the maximum packet length of a flow or arrival process typically does not change in the network. For this reason as well as the discussion above, the maximum packet length parameter is not included in Definition~\ref{def-tspec}. 

Third, for flows characterized by this TSpec, few results are available for their delay guarantee analysis. 
On the contrary, a rich set of such results have already been developed in NC, e.g. \cite{Cruz91a}\cite{Cruz91b}\cite{Chang00}\cite{NetCal}\cite{Sigcomm06}\cite{SNC}\cite{Jiang-GL09}\cite{Jiang-mascots09}\cite{Jiang-valuetools11}\cite{Liebeherr17}. So, an idea is to find a way to link TSN / DetNet TSpec to traffic models in NC, though this traffic specification is fundamentally different. 

In the following, we introduce a traffic model that is related to NC, and prove its relationship with the TSN / DetNet TSpec. 

\subsection{The Max-Plus Traffic Model and the Mapping}

In this paper, we introduce the following traffic model. 
\begin{definition}
An arrival process\nop{, characterized by $\bar{A}(n)$,} is said to be $(\lambda,  \nu)$-constrained, if, for all $n \ge m \ge 0$, there holds
$$
\bar{A}(m, n) \ge \frac{1}{\lambda}  (n-m - \nu)^{+}
$$
where $(x)^{+} \equiv \max\{x, 0 \}$ and $\lambda (> 0)$ and $\nu (\ge 0)$ are two constant parameters. 
\end{definition}

As the definition shows, the  $(\lambda,  \nu)$ model is defined on the arrival time function. Indeed, it is a special case of the max-plus arrival curve model defined for the max-plus network calculus  \cite{Chang00} \cite{Jiang-GL09} \cite{Jiang-mascots09}, where a more general function, called max-plus arrival curve, is used as the constraint function. 

The following lemma shows that, the definition of the $(\lambda,  \nu)$ model is equivalent to an expression in the max-plus algebra, and is hence referred to as {\em a max-plus traffic model}. The proof is similar to that for the general max-plus arrival curve model in Lemma 5.2 in \cite{Jiang-mascots09} and omitted.

\begin{lemma}
If an arrival process is $(\lambda,  \nu)$-constrained, if and only if, there holds
$$
\bar{A}(n) \ge \bar{A} \bar{\otimes} \bar{\alpha}(n)
$$
where $\bar{\alpha}(n) = \frac{1}{\lambda}  (n-m - \nu)^{+}$, and the operation $\bar{\otimes}$ of two functions $F(n)$ and $G(n)$ is the max-plus convolution, defined as $F \bar{\otimes} G (n) \equiv \sup_{0 \le m \le n} \{F(m) + G(n-m)\}$. 
\end{lemma}

The following theorem establishes a relationship\nop{ the mapping} between the TSN/DetNet TSpec and the $(\lambda,  \nu)$ model.
\begin{theorem}\label{th-0}
(i) If an arrival process is $(\lambda,  \nu)$-constrained, it conforms to the TSN / DetNet traffic specification with 
\begin{itemize}
\item[(a)] interval parameter $\tau=j/\lambda$ and maximum packet number parameter $K =\lceil \nu \rceil +j+1$, or 
\item[(b)] interval parameter $\tau={(j/\lambda)}_{-}$ and maximum packet number parameter $K =\lceil \nu \rceil+j$, 
\end{itemize}
for any integer $j \ge 1$, where ${x}_{-}$ denotes $x-\epsilon$ for $\epsilon \to 0$. \\
(ii) If an arrival process conforms to the TSN / DetNet traffic specification with parameters $\tau$ and $K (\ge 1)$, it is $(\lambda,  \nu)$-constrained with $\lambda=K/\tau$ and $\nu=K-1$.
\end{theorem}

\begin{proof}
\nop{
For the first part, the condition implies, for any $m \ge 1$
$$
\bar{A}(m, m+\nu +2) \ge {A}(m, m+\nu+2) \ge \frac{2}{\lambda} > \tau. 
$$
which says the time distance between any two packets that are $\nu+2$ apart is greater than $\tau$.  In other words, such two packets cannot be in an interval of length $\tau$. Equivalently, this is to say that in an interval of length $\tau$, the maximum number of packets cannot exceed $\nu +2$, which proves the first part. 
}

For the first part, the condition implies, for any $m \ge 1$ and for $\forall j \ge 1$, 
$$
\bar{A}(m, m+\lceil \nu \rceil +j+1) \ge \frac{\lceil \nu \rceil + j+1 - \nu}{\lambda} \\
\ge \frac{j+1}{\lambda} > \frac{j}{\lambda}. 
$$
This is to say the time distance between any two packets that are $\lceil \nu \rceil +j+1$ apart is greater than $\frac{j}{\lambda}$.  In other words, such two packets cannot be in an interval of length $\frac{j}{\lambda}$. Equivalently, this is to say that in an interval of length $\frac{j}{\lambda}$, the maximum number of packets cannot exceed $\lceil \nu \rceil +j +1$.\footnote{Without loss of generality, suppose packet $m$ is the first packet in the period. Note that from packet $m$ to packet $m+ \lceil \nu \rceil +j +1$, there are in total $\lceil \nu \rceil +j +2$ packets. However, since $\bar{A}(m, m+\lceil \nu \rceil +j+1) >  \frac{j}{\lambda}$, the last packet, i.e. packet $m+\lceil \nu \rceil +j +1$, cannot be within this period. So, the total number of packets in this period will not exceed $\lceil \nu \rceil +j +1$.} 

Indeed, for the first part, we also have 
$$
\bar{A}(m, m+ \lceil \nu \rceil +j) \ge \bar{A}(m, m+\nu+j) \ge \frac{j}{\lambda} > \left( \frac{j}{\lambda} \right)_{-}.
$$
Similarly, this is to say that in an interval of length $\left( \frac{j}{\lambda} \right)_{-}$, the maximum number of packets does not exceed $\lceil \nu \rceil +j$. 

For the second part, under the given condition, we have 
\begin{eqnarray}
\bar{A}(m, n) &\ge& \left\lfloor \frac{n-m}{K} \right\rfloor \tau = \left\lceil \frac{n-m-K+1}{K} \right\rceil K \lambda^{-1} \nonumber \\
&\ge& \left( \frac{n-m-K+1}{K} \right)^{+} K  \lambda^{-1}  \nonumber \\
&=&   \lambda^{-1} (n-m-K+1)^{+} \nonumber
\end{eqnarray}
which concludes the second part.
\end{proof}

{\bf Remarks:} From the second half of Theorem 1.(i), if $\nu$ is an integer and let $j=1$, we then obtain that if an arrival process is $(\lambda,  \nu)$-constrained, it conforms to the TSN / DetNet traffic specification with parameters $\tau = \left( \frac{1}{\lambda} \right)_{-}$ and $K= \nu +1$. Here the mapping between $\nu$ and $K$ in the two models is the same as from the second part of the theorem, i.e. Theorem 1.(ii). However, it is worth highlighting that for parameters $\lambda$ and $\tau$, the relation $\tau = \left( \frac{1}{\lambda} \right)_{-}$ from the max-plus traffic model to the TSN / DetNet TSpec is no more recovered from the reverse relation from Theorem 1.(ii) where we differently have $\lambda=K/\tau$. This implies that the two models are in general not equivalent to each other.

\subsection{The Analogy Min-Plus $(\sigma, \rho)$ Traffic Model}

The well-known $(\sigma, \rho)$ traffic model is as the following \cite{Cruz91a}: 
\begin{definition}
An arrival process\nop{, characterized by ${A}(t)$, } is said to be $(\sigma, \rho)$-constrained, if, for all $t \ge s \ge 0$,\nop{ the cumulative amount of traffic $A(s,t)$ in $[s,t]$ satisfies}
$$
A(s,t) \le \rho t + \sigma
$$
where parameters $\rho (> 0)$ and $\sigma (\ge 0)$ are often called the rate and burst parameters respectively.
\end{definition}

It is also known (see e.g. \cite{NetCal}) that the definition of the $(\sigma, \rho)$ model is equivalent to the following, and hence referred to as {\em a min-plus traffic model}:
\begin{lemma}
An arrival process is $(\sigma, \rho)$-constrained, if and only if, there holds
$$
A(t) \le A \otimes \alpha(t)
$$
where $\alpha(t) = \rho t + \sigma$, and the operation $\otimes$ of two functions $F(n)$ and $G(n)$ is the min-plus convolution, defined as $F {\otimes} G (t) \equiv \inf_{0 \le s \le t} \{F(s) + G(t-s)\}$. 
\end{lemma}

Note that for any period defined by $s (\le t)$ and $t$, we always have $A(s,t) = \sum_i^{I} A_i(s,t)$, based on which, the superposition property of the  $(\sigma, \rho)$ model is easily verified (see e.g. \cite{NetCal}):

\begin{lemma}\label{lm-sr}
 Consider the superposition of $I (\ge 2)$ arrival processes ${A}_i(t)$, $i = 1, \dots, I$. If each arrival process $A_i(t)$ is $(\sigma_i, \rho_i)$-constrained, the aggregate process ${A}(t)$ is  $(\sigma, \rho)$-constrained with 
 $$\rho = \sum_{i=1}^{I} \rho_i; \qquad \sigma = \sum_{i=1}^{I}\sigma_i.$$
\end{lemma} 

In contrast to the min-plus $(\sigma, \rho)$ model, for the max-plus $(\lambda, \nu)$ model, its superposition property has not been found / proved. In fact, the superposition property of the more general arrival curve model is a long-standing open problem \cite{Chang00}\cite{Jiang-mascots09}\cite{Liebeherr17}. This motivates and is focused in the next section.  

\section{The Superposition Property of the Max-Plus Traffic Model} \label{sec-3}
 
\subsection{The Difficulty}
For the superposition of arrival processes, the following relationship was initially derived in \cite{Jiang-mascots09} and has also been verified in \cite{Liebeherr17}:  

\begin{lemma}
Given the arrival time function $\bar{A}_i(n)$ of each individual process, the arrival time function $\bar{A}(n)$ of the aggregate process can be related to $\bar{A}_i(n)$ as, 
\begin{equation} \label{aggregate}
\bar{A}(n) = \inf_{m_1 + \cdots + m_I=n} \max_{i=1, \dots, I}\bar{A}_i(m_i) .
\end{equation}
\end{lemma}

The expression (\ref{aggregate}) is neat, based on which, we can write 
\begin{eqnarray}
\bar{A}(m,n) &=&  \inf_{m_1 + \cdots + m_I=n} \max_{i=1, \dots, I}\bar{A}_i(m_i) - \nonumber\\
&& \inf_{m_1 + \cdots + m_I=m} \max_{i=1, \dots, I}\bar{A}_i(m_i) \label{agg-2} 
\end{eqnarray}

Unfortunately, it is unknown how to further relate the right hand side of (\ref{agg-2}) directly to $\bar{A}_i(m_i,n_i)$, i.e. to write the right hand side as a function of and only of $\bar{A}_i(m_i,n_i)$, $i=1, \dots, I$. This makes it difficult to find the superposition property of the $(\lambda, \nu)$ model from the above relationship. 

To bypass this difficulty, when packet length information is known, an indirect approach (see e.g., \cite{Chang00} \cite{Jiang-mascots09}) has been proposed. While this indirect approach is mathematically sound, its application is limited, some compromise may have to be made and the result can be loose. More discussion on these will be provided in Sec. \ref{sec-4}.
 
\subsection{The Superposition Property of the $(\lambda,  \nu)$ Model}  

This subsection is devoted to finding and proving the superposition property of the arrival time function based $(\lambda,  \nu)$ max-plus traffic model, summarized in the following theorem.  

\begin{theorem}\label{th-1}
 Consider the superposition of $I (\ge 2)$ arrival processes $\bar{A}_i$, $i = 1, \dots, I$. If all arrival processes $\bar{A}_i$ are $(\lambda_i, \nu_i)$-constrained,  the aggregate process $\bar{A}$ is $(\lambda^{dir.},  \nu^{dir.})$-constrained with
 $$
\lambda^{dir.} = \sum_{i=1}^{I} \lambda_i; \qquad
\nu^{dir.} = \sum_{i=1}^{I}\nu_i + (I-1). 
 $$
 \end{theorem}

Theorem \ref{th-1} can be proved by induction. 
 We first present the base case with $I=2$ in Lemma~\ref{lm-1} and its proof.  
 
 \begin{lemma}\label{lm-1}
 Consider the superposition of two processes $\bar{A}_i$, $i = 1, 2$. If both processes $\bar{A}_i$ are $(\lambda_i, \nu_i)$-constrained,  the aggregate process $\bar{A}$ is $(\lambda,  \nu)$-constrained with
 $$
\lambda =  \lambda_1 + \lambda_2; \qquad \nu = \nu_1 + \nu_2 + 1.
 $$
 \end{lemma}

 \begin{proof}
Though lengthy, the complete proof is provided below, as we believe, the techniques used in the proof also provide insights when dealing with similar problems. In addition, the proof itself also serves as an indication of the difficulty as discussed in the previous subsection. 

To help the presentation, we let 
$$\bar{\alpha}(n) = \frac{1}{\lambda}  (n - \nu)^{+} =   \frac{1}{\lambda_1+\lambda_2}  (n - \nu_1 - \nu_2 -1)^{+} .$$ 
Then, with the definition of the $(\lambda,  \nu)$ model, to prove the lemma is to prove that, for all $n \ge m \ge 0$, there holds:
 \begin{equation}\label{eq-lm-1-0}
 \bar{A}(m, n) \ge \bar{\alpha}(n-m).
 \end{equation}

We start with two trivial cases. One is, for any $n=m (\ge 0)$, $\bar{A}(m, n) = 0$ by definition, with which, $(\ref{eq-lm-1-0})$ holds because $ \bar{\alpha}(0) =( - \nu)^{+} = 0$. Another is, for any $n>m (\ge 0)$ with $n-m=1$, $\bar{A}(m, n) \ge 0$ because of non-negative inter-arrival time between $m$ and $m+1$, with which, $(\ref{eq-lm-1-0})$ holds because $ \bar{\alpha}(1) = ( - \nu_1 - \nu_2)^{+} = 0$. 

Next, we consider any $n > m (\ge 0)$ with $n-m >1$. The corresponding time period is $[\bar{A}(m), \bar{A}(n)]$. We denote the set of packets {\bf between} $m$ and $n$ in $\bar{A}$ as $\{m+1, \dots, n-1\}_{\bar{A}}$. \footnote{This set has been intentionally used in the proof to avoid ambiguity that would arise if the time period $[\bar{A}(m), \bar{A}(n)]$ had been used, because concurrent arrivals may exist or happen both in the individual arrival processes and in the aggregate process even at $\bar{A}(m)$ and/or $\bar{A}(n)$, which cannot be distinguished by using $[\bar{A}(m), \bar{A}(n)]$.}

Without loss of generality, we suppose customer $n$ is from $\bar{A}_1$ and is the $n_1$-th customer in $\bar{A}_1$. In other words, we have 
 \begin{equation}
\bar{A}(n)=\bar{A}_1(n_1). \label{eq-ft-0}
 \end{equation}
Under this setting, there are three possibilities about customer $A(m)$: ({\bf Case 1}) It is either from $\bar{A}_1$, or ({\bf Case 2}) is from $\bar{A}_2$, or ({\bf Case 3}) is the virtual packet at time 0 for which we have $A(0)=0$. {\em For the first two cases, we must have $m\ge1$, and for the third case, $m=0$}. Accordingly, we prove for the three cases:  

{\bf Case 1: Packet $m$ in the aggregate process is from $\bar{A}_1$.} Let $m_1$ denote its number in $\bar{A}_1$, which implies: \begin{eqnarray}
\bar{A}(m) &=& \bar{A_1}(m_1)  \\
\bar{A}(m, n) &=& \bar{A_1}(m_1, n_1) 
\end{eqnarray}

Now, given $m$ and $n$ are both from $\bar{A}_1$, there are (and only) three sub-cases, Case 1.1 - Case 1.3, which we consider below.

{\em Case 1.1: In $\{m+1, \dots, n-1\}_{\bar{A}}$, there is no packet from $\bar{A}_2$.}  In this sub-case, we have: 
\begin{eqnarray}
n - m &=& n_1-m_1. \label{eq-ft-1}
\end{eqnarray}
In addition, since  $\bar{A}_1$ is constrained  by $(\lambda_1, \nu_1)$, we have 
\begin{eqnarray}
(\lambda_1 + \lambda_2) \cdot \bar{A}(m, n)  &\ge& \lambda_1 \cdot \bar{A}(m, n)  = \lambda_1 \cdot \bar{A_1}(m_1, n_1) \nonumber \\
 &\ge& (n_1 - m_1 -  \nu_1)^{+} \label{eq-tm1} \nonumber\\
 &=& (n - m - \nu_1)^{+} \nonumber
\end{eqnarray}
which gives 
\begin{equation}
\bar{A}(m, n)  \ge \frac{1}{\lambda_1 + \lambda_2} (n-m -  \nu_1)^{+}  \ge \bar{\alpha}(n-m). \nonumber
\end{equation}

{\em Case 1.2: In $\{m+1, \dots, n-1\}_{\bar{A}}$, there is one packet from $\bar{A}_2$.} In this sub-case, we have: 
\begin{eqnarray}
n - m &=& (n_1-m_1) + 1 \label{eq-ft-2}
\end{eqnarray}
where, on the right hand side, the first term represents the number of intervals in $\bar{A}_1$ and the second term represents that an additional interval is introduced because of the one packet from $\bar{A}_2$, in $\{m, \dots, n\}_{\bar{A}}$. 

Similarly, we have
\begin{eqnarray}
(\lambda_1 + \lambda_2) \cdot \bar{A}(m, n)  &\ge& \lambda_1 \cdot \bar{A_1}(m_1, n_1) \ge (n_1 - m_1) - \nu_1 \nonumber \\
 &=& (n - m -1 - \nu_1)^{+} \nonumber
\end{eqnarray}
which gives 
\begin{equation}
\bar{A}(m, n)  \ge \frac{1}{\lambda_1 + \lambda_2} \cdot (n-m -  \nu_1 -  1)^{+} \ge \bar{\alpha}(n-m). \nonumber
\end{equation}

{\em Case 1.3: In $\{m+1, \dots, n-1\}_{\bar{A}}$, there are multiple packets from $\bar{A}_2$.} Without loss of generality, let $m_2$ be the first and $n_2$ be the last of these packets from $\bar{A}_2$. In this sub-case, the following facts hold: 
\begin{eqnarray}
\bar{A}(n) &\ge& \bar{A}_2 (n_2) \\
\bar{A}(m) &\le& \bar{A}_2 (m_2) 
\end{eqnarray}
which gives $\bar{A}(m, n)  \ge \bar{A_2}(m_2, n_2) $. 
In addition, we have
\begin{eqnarray}
n-m &=& (n_1 - m_1) + (n_2 - m_2) +1 \label{eq-ft-3}
\end{eqnarray}
where the left hand side represents the number intervals between packets $m$ and $n$ in $\bar{A}$. For the right hand side, in $\{m, \dots, n\}_{\bar{A}}$, we now have $(n_1 - m_1 +1)$ packets from $\bar{A}_1$, and $(n_2 - m_2 +1)$ packets from $\bar{A}_2$, which in total gives $(n_1 - m_1) + (n_2 - m_2) +2 \equiv N$ number of packets that have $N - 1$ intervals, which is $(n_1 - m_1) + (n_2 - m_2) +1$. 

We then have
\begin{eqnarray}
&& (\lambda_1 + \lambda_2) \cdot \bar{A}(m, n)  \nonumber\\
&=& \lambda_1 \cdot \bar{A_1}(m_1, n_1) +   \lambda_2 \cdot \bar{A}(m, n) \nonumber \\
&\ge& \lambda_1 \cdot \bar{A_1}(m_1, n_1) +   \lambda_2 \cdot \bar{A_2}(m_2, n_2) \nonumber \\
&\ge& (n_1 - m_1 - \nu_1)^{+} + (n_2 - m_2 -  \nu_2)^{+}  \nonumber \\
&\ge& ((n_1 - m_1 - \nu_1) + (n_2 - m_2 -  \nu_2))^{+} \nonumber \\
&=&  ((n - m-1) - ( \nu_1 +  \nu_2))^{+} 
\end{eqnarray}
and hence
\begin{eqnarray}
\bar{A}(m, n)  &\ge& \frac{1}{\lambda_1 + \lambda_2} (n-m -  \nu_1 - \nu_2 -1)^{+}  = \bar{\alpha}(n-m). \nonumber
\end{eqnarray}

Combing\nop{ the three sub-cases,} Case 1.1 - Case 1.3, (\ref{eq-lm-1-0}) is proved for the first case. In the following, we consider the second case. 

{\bf Case 2: Packet $m$ in the aggregate process is from $\bar{A}_2$.}  Without of generality, suppose it is the $m_2$-th packet in $\bar{A}_2$, which also implies 
\begin{eqnarray}
\bar{A}_2(m_2) &=&\bar{A}(m)  \label{c2-tm-1}
\end{eqnarray}

In this case, there are also (and only) three sub-cases, Case 2.1 - Case 2.3, which we consider below. 

{\em Case 2.1: In $\{m+1, \dots, n-1\}_{\bar{A}}$, there is no packet from $\bar{A}_1$ but there is at least one packet from $\bar{A}_2$.}  Let $n_2$ denote the last such packet from $\bar{A}_2$. Based on the definition of $n_2$, we must have 
\begin{eqnarray}
\bar{A_2}(n_2) &\le& \bar{A}(n) = \bar{A_1}(n_1) \label{c2-tm-2} \\
n - m &=& (n_2 - m_2) + 1 \label{eq-ft-5}
\end{eqnarray}
where, on the right hand side of (\ref{eq-ft-5}), the first term $(n_2-m_2)$ represents the number of intervals of packets from $\bar{A}_2$ and the second term represents the additional interval introduced by the one packet, i.e. $n_1$, from $\bar{A}_1$ in  $\{m, \dots, n\}_{\bar{A}}$. 

With (\ref{c2-tm-1}) and (\ref{c2-tm-2}), we now have,
\begin{eqnarray}
&& (\lambda_1 + \lambda_2) \cdot \bar{A}(m, n)  \nonumber\\
&\ge&  \lambda_2 \cdot (\bar{A}(n) - \bar{A}(m) ) \ge  \lambda_2 \cdot (\bar{A_2}(n_2) - \bar{A_2}(m_2) )\nonumber \\
&\ge& (n_2 - m_2 -  \nu_2 )^{+}=  (n - m -1 - \nu_2 )^{+} \nonumber
\end{eqnarray}
and hence
\begin{eqnarray}
\bar{A}(m, n)  &\ge& \frac{1}{\lambda_1 + \lambda_2} (n-m -  \nu_2 -1)^{+}  \ge \bar{\alpha}(n-m). \nonumber 
\end{eqnarray}

{\em Case 2.2: In $\{m+1, \dots, n-1\}_{\bar{A}}$, there is no packet from $\bar{A}_2$ but there is at least one packet from $\bar{A}_1$.}  Let $m_1$ denote the first such packet from $\bar{A}_1$. Based on the definition of $m_1$, we must have 
\begin{eqnarray}
\bar{A_1}(m_1) &\ge& \bar{A}(m) = \bar{A_2}(m_2) \label{c2-tm-3} \\
n - m &=& (n_1 - m_1) + 1 \label{eq-ft-6}
\end{eqnarray}
where, on the right hand side of (\ref{eq-ft-6}), the first term $(n_1-m_1)$ represents the number of intervals of packets from $\bar{A}_1$ and the second term represents that an additional interval is introduced by the one packet, i.e. $m_2$, from $\bar{A}_2$ in  $\{m, \dots, n\}_{\bar{A}}$.

With (\ref{eq-ft-0}),  (\ref{c2-tm-3}) and (\ref{eq-ft-6}), we now have,
\begin{eqnarray}
&& (\lambda_1 + \lambda_2) \cdot \bar{A}(m, n)  \nonumber \\
&\ge&  \lambda_1 \cdot (\bar{A}(n) - \bar{A}(m) ) \ge  \lambda_1 \cdot (\bar{A_1}(n_1) - \bar{A_1}(m_1) )\nonumber \\
&\ge& (n_1 - m_1 -  \nu_1)^{+} =  (n - m -1 - \nu_1)^{+}  \nonumber
\end{eqnarray}
and hence
\begin{eqnarray}
\bar{A}(m, n)  &\ge& \frac{1}{\lambda_1 + \lambda_2} (n-m - \nu_1 - 1)^{+}  \ge \bar{\alpha}(n-m). \nonumber
\end{eqnarray}

{\em Case 2.3: In $\{m+1, \dots, n-1\}_{\bar{A}}$, there is at least one packet from $\bar{A}_1$ and there is at least one packet from $\bar{A}_2$.}  Let $m_1$ denote the first such packet from $\bar{A}_1$, and $n_2$ the last such packet from $\bar{A}_2$. Based on the definitions of $m_1$ and $n_2$, we must have 
\begin{eqnarray}
\bar{A_1}(m_1) &\ge& \bar{A}(m) = \bar{A_2}(m_2) \label{c2-tm-3-a} \\
\bar{A_2}(n_2) &\le& \bar{A}(n) = \bar{A_1}(n_1) \label{c2-tm-2-a}  \\
n - m &=& (n_1 - m_1) + (n_2- m_2)  +1 \label{eq-ft-7}
\end{eqnarray}
where (\ref{c2-tm-3-a}) is the same as (\ref{c2-tm-3}),  (\ref{c2-tm-2-a}) the same as (\ref{c2-tm-2}), and on the right hand side of (\ref{eq-ft-7}), the first term $(n_1-m_1)$ represents the number of intervals of packets from $\bar{A}_1$,  the second term $(n_2-m_2)$ represents the number of intervals of packest from $\bar{A}_2$, and the third term represents that an additional interval needs to be added due to the superposition, all in $\{m, \dots, n\}_{\bar{A}}$. (See also the discussion for (\ref{eq-ft-3}).)

With (\ref{c2-tm-3-a}) and (\ref{c2-tm-2-a}), we now have,
\begin{eqnarray}
&& (\lambda_1 + \lambda_2) \cdot \bar{A}(m, n)  \nonumber \\
&=&  \lambda_1 \cdot (\bar{A}(n) - \bar{A}(m) ) +  \lambda_2 \cdot (\bar{A}(n) - \bar{A}(m) ) \nonumber \\
&\ge&  \lambda_1 \cdot (\bar{A_1}(n_1) - \bar{A_1}(m_1) ) + \lambda_2 \cdot (\bar{A}(n_2) - \bar{A}(m_2) )\nonumber \\
&\ge& (n_1 - m_1 -  \nu_1)^{+} +  (n_2 - m_2 -  \nu_2)^{+} \nonumber \\
&\ge& ((n_1 - m_1) -  \nu_1 +  (n_2 - m_2) -  \nu_2)^{+} \nonumber \\
&=&  (n - m-1 - \nu_1- \nu_2)^{+}
\end{eqnarray}
and hence
\begin{eqnarray}
\bar{A}(m, n)  &\ge& \frac{1}{\lambda_1 + \lambda_2}(n-m - \nu_1 - \nu_2 - 1)^{+}  = \bar{\alpha}(n-m). \nonumber 
\end{eqnarray}

Combing Case 2.1 - Case 2.3, (\ref{eq-lm-1-0}) is proved for the second case. With this, we have proved (\ref{eq-lm-1-0}) holds for all $n > m \ge 1$. 

{\bf Case 3: Customer $m$ is the virtual packet at the origin, i.e. $m=0$ and $A(0)=0$.}  In this case, in addition to the $n_1$ customers from $\bar{A}_1$, there are $n- n_1$ customers from $\bar{A}_2$ in the period, and we must also have 
$$\bar{A}_2(n-n_1) \le \bar{A}_1(n_1) = \bar{A}(n)$$  
with which, we further obtain
\begin{eqnarray}
&& (\lambda_1 + \lambda_2) \cdot \bar{A}(0, n)  \nonumber \\
&=&  \lambda_1 \cdot \bar{A}(n)  +  \lambda_2 \cdot \bar{A}(n) \nonumber \\
&\ge&  \lambda_1 \bar{A_1}(n_1) + \lambda_2 \cdot \bar{A}_2(n -n_1) \nonumber \\
&\ge& (n_1 -  \nu_1)^{+} +  ((n- n_1) -  \nu_2)^{+} \nonumber \\
&\ge& ((n_1 -  \nu_1) +  ((n- n_1) -  \nu_2))^{+}  = (n  - \nu_1- \nu_2)^{+} \nonumber 
\end{eqnarray}
and hence
\begin{eqnarray}
\bar{A}(0, n)  &\ge& \frac{1}{\lambda_1 + \lambda_2}((n-0) - \nu_1 - \nu_2 )^{+}   \ge \bar{\alpha}(n-0). \nonumber 
\end{eqnarray}
This, together with the proof for Case 1 and Case 2, ends the proof of Lemma \ref{lm-1}. 
\end{proof} 

Next for the induction, we prove Theorem \ref{th-1} also holds for $I + 1$ arrival processes, given the condition that it holds for $I$ arrival processes. 
Note that, under the given condition, the aggregate process of $I$ arrival processes is $(\lambda_{(I)},  \nu_{(I)})$-constrained with
 $
\lambda_{(I)} = \sum_{i=1}^{I} \lambda_i; 
\nu_{(I)} = \sum_{i=1}^{I}\nu_i + (I-1). 
 $

The aggregate process of $I+1$ arrival processes, denoted as $\bar{A}_{(I+1)}$, can be treated as the superposition of two processes $\bar{A}_{(I)}$ and $\bar{A}_{I+1}$, where $\bar{A}_{(I)}$ denotes the aggregate of the first $I$ processes and $\bar{A}_{I+1}$ the last process. Then, with Lemma \ref{lm-1}, $\bar{A}_{(I+1)}$ is $(\lambda_{(I+1)},  \nu_{(I+1)})$-constrained with
\begin{eqnarray}
\lambda_{(I+1)} &=& \lambda_{(I)}  +\lambda_{I+1} = \sum_{i=1}^{I+1} \lambda_i \nonumber \\
\nu_{(I+1)} &=& \nu_{(I)} + \nu_{I+1} +1 = \sum_{i=1}^{I+1}\nu_i + ((I+1)-1) \nonumber
\end{eqnarray}
which is Theorem \ref{th-1} for the superposition of  $I+1$ processes. This completes the proof of Theorem \ref{th-1} .

\subsection{Extensions}

It is worth highlighting that the superposition property of the $(\lambda, \nu)$  model presented in Theorem \ref{th-1} resembles that of the $(\sigma, \rho)$ model shown in Lemma \ref{lm-sr}. 

Following the essence in the proof of Theorem \ref{th-1}, the following superposition property can be proved for the TSN / DetNet traffic specification. 

\begin{corollary}\label{cor-tsn}
 Consider the superposition of $I (\ge 2)$ arrival processes $\bar{A}_i$, $i = 1, \dots, I$. If each arrival process $\bar{A}_i$ confirms to the TSN / DetNet traffic specification with interval $\tau_i$ and maximum packet number $K_i$, then the aggregate process $\bar{A}$ also confirms to the TSN / DetNet traffic specification with interval $\tau$ and maximum packet number $K$ where
 $$
\tau^{-1} = \sum_{i=1}^{I} \tau_i^{-1}; \qquad
K = \sum_{i=1}^{I}K_i  . 
 $$
 \end{corollary}

In addition, the superposition property of the $(\lambda, \nu)$  model can be extended to the more general max-plus arrival curve model shown below. 

\nop{
\subsection{Tightness}
\begin{theorem}
For the superposition of $I$ arrival processes, if each of them is $(\lambda, \nu)$-constrained with $\bar{A}(n-m) =\frac{1}{\lambda_i}(n-m-\nu_i)^{+}$, where $\nu_i$ is a non-negative integer, then there exists an superposition case where for some $n \ge m - \nu -I +1$, $\bar{A}(n-m) = \frac{1}{\lambda}(n-m-\nu-I+1)^{+}$ for the aggregate process.
\end{theorem}
}

\begin{corollary}
Consider the superposition of $I (\ge 2)$ arrival processes $\bar{A}_i$, $i = 1, \dots, I$. If each of them has a max-plus arrival curve $\bar{\alpha}_i(\cdot) (\ge 0)$, $(i=1, \dots, I)$, then the superposition process $\bar{A}$ has a max-plus arrival curve $\bar{\alpha}$ as  
 $$
 \bar{\alpha}(n) = \frac{1}{\lambda} \cdot (n -\nu)^{+} 
 $$ 
where 
$$
\lambda = \sum_{i=1}^{I} \lambda_i; \qquad \nu = \sum_{i=1}^{I}  \nu_i + (I-1) 
$$
with 
\begin{eqnarray}
\lambda_i &=& \sup\{r: r \cdot \bar{\alpha}_i(n) \le n \} \label{lambda_i}\\
\nu_i &=& n - \lambda_i \bar{\alpha}_i(n) .  
\end{eqnarray}
\end{corollary}

\begin{proof}
Under the given assumptions, we can re-write $\bar{\alpha}_i(n)$ as 
$$
\bar{\alpha}_i(n) = \frac{1}{\lambda_i} (n - \nu_i)^{+}. 
$$
Then, the result follows immediately from Theorem \ref{th-1} \footnote{With (\ref{lambda_i}), i.e. the definition of $\lambda_i$, $\nu_i $ is non-negative in nature.}. 
\end{proof}

Furthermore, it is worth highlighting that, with the help of Theorem \ref{th-0} and Theorem \ref{th-1}, the existing NC results can be made use of for delay guarantee analysis of TSN / DetNet.

\section{Comparison}\label{sec-4}

In this section, we compare the superposition results obtained using the indirect approach and those using the direct approach proposed in the literature (see e.g., \cite{Chang00} \cite{Jiang-mascots09}). 

Specifically, the indirect approach first transforms the $(\lambda,  \nu)$ characterization from the arrival time function to the $(\sigma, \rho)$ traffic characterization, then applies the superposition property of the $(\sigma, \rho)$ model to find the $(\sigma, \rho)$ characterization for the aggregate process, and finally transforms the obtained  $(\sigma, \rho)$ characterization back to the the $(\lambda,  \nu)$ characterization. 

The following lemma summarizes the result from the indirect approach. Its proof is omitted, since a general but much more complex form can be found from Corollary 6.2.9 in \cite{Chang00}. 

\begin{lemma}\label{lm-indirect}
Consider the superposition of $I (\ge 2)$ arrival processes $\bar{A}_i$, $i = 1, \dots, I$. If each $\bar{A}_i$ is $(\lambda_i, \nu_i)$-constrained with maximum packet length $l_i$ and the minimum packet length of all processes is known, denoted as $l$, then the aggregate process $\bar{A}$ is $(\lambda^{ind.},  \nu^{ind.})$-constrained with
 $$
\lambda^{ind.} = \sum_{i=1}^{I} \frac{l_i}{l} \lambda_i; \qquad
\nu^{ind.} = \sum_{i=1}^{I}(\nu_i +1) \frac{l_i}{l}.
 $$
\end{lemma}

Comparing Lemma \ref{lm-indirect} with Theorem \ref{th-1}, in addition to how their results are derived, there are {\em two fundamental differences}: 
\begin{itemize}
\item For Lemma \ref{lm-indirect} to be applicable, we at least need to know the maximum packet length of each process and the minimum packet length of all processes. In contrast, no specific packet length information is required for Theorem~\ref{th-1}. This difference has an immediate consequence, which is, if the packet length information is not known or provided, the superposition result presented in Lemma \ref{lm-indirect}  can no more be used. 

\item Even when the needed packet length information condition for Lemma \ref{lm-indirect} to be applicable is available, its resultant $(\lambda,  \nu)$ representation is worse than what is from Theorem \ref{th-1}. This is because $\lambda^{ind.} \ge \lambda^{dir.}$ and $\nu^{ind.} > \nu^{dir.}$ leading to a smaller or worse bounding function $\lambda^{-1} (n-m-\nu)$ in the $(\lambda,  \nu)$ characterization. 
\end{itemize}

In the rest, we presents results for four extremely simple cases to exemplify the comparison. For simplicity in the expression, we assume every flow $i$ produces packets periodically and the period length is $\tau_i$. In addition, for ease of expression, we consider the superposition of only two flows, i.e. $I=2$. 

The other settings of the four cases are:
\begin{itemize}
\item Case 1: All flows have the same period $\tau_i=\tau$.
\item Case 2: All flows have the same period $\tau_i=\tau$ and the same packet length $l_i=l$. 
\item Case 3: All flows still have the same packet length $l_i=l$, but while  one flow has period $\tau_1=\tau$,  the other flow has period $\tau_2=2\tau$. 
\item Case 4: All other settings are the same as for the second case, except that the second flow has packet length $l_2=2l$. As a remark, in this case, the average traffic rate (in bps) of the second flow is the same as that of the first flow, i.e. $\rho_1=\rho_2 = l/\tau$.
\end{itemize}

Table \ref{tb-comp} summarizes and compares the superposition results from both approaches for the four cases. Though simple, the comparison validates the discussion about the  fundamental differences between the indirect and direct approaches. 

\begin{table}[t!] 
\caption{Comparison of superposition property results}
\label{tb-comp} 
\begin{center}
\begin{tabular}{ |c||c|c| } 
\hline
Cases: & Indirect Appr. (Lemma  \ref{lm-indirect})  & Direct Appr. (Theorem \ref{th-1}) \\ 
 \hline
Case 1 & Not Available & $\frac{\tau}{2}(n - 1)^{+}$ \\ 
 \hline
Case 2 & $\frac{\tau}{2}(n - 2)^{+}$ & $\frac{\tau}{2}(n - 1)^{+}$ \\ 
 \hline
Case 3 & $\frac{2\tau}{3}(n - 2)^{+} $ & $\frac{2\tau}{3}(n - 1)^{+}$ \\ 
 \hline
Case 4 & $\frac{\tau}{2}(n - 3)^{+}$  & $\frac{2\tau}{3}(n - 1)^{+}$ \\
 \hline
\end{tabular}%
\end{center}
 \end{table}

\section{Conclusion}\label{sec-5}
The emerging time-sensitive networking (TSN) and deterministic networking (DetNet) standards (re-)call attention to the network calculus, in order to make use of the rich set of results available in NC. In this paper, we introduced an arrival time function based max-plus NC traffic model. We proved that it is closely related to the TSN TSpec and there is a direct mapping between them. In addition, another focus has been on finding and proving the superposition property of the max-plus traffic model, providing answer to a long-standing question in the max-plus network calculus. The proof adopted a novel direct approach that requires no packet length information, in contrast to a literature indirect approach. Appealingly, the proved superposition property shows clear analogy with that of the well-known counterpart $(\sigma, \rho)$ model in NC. The comparison of the superposition results from the indirect and direct approaches not only shows wider applicability of the superposition property obtained in this paper, but also offers better traffic characterization for the aggregate process. These results can help  make use of the NC results for delay guarantee analysis of TSN / DetNet networks.

\section*{Acknowledgment}
This is an updated version. The initial version of this paper was submitted to IEEE Globecom 2018 and will be presented there. The author would like to thank its anonymous reviewers for their helpful comments, and Jean-Yves Le Boudec for similar comments. It is mainly based on those comments that this updated version has been produced. In addition, the author would like to specially thank Jean-Yves Le Boudec for pointing out that there is an equivalent model of the $(\lambda, \nu)$ model, which is called ``packet burstiness'' constraint PB$(\rho, K)$ independently introduced in \cite{LeBoudec18}, and that based on the PB model and results in \cite{LeBoudec18}, a simplified proof of the superposition property for the $(\lambda, \nu)$ model may be obtained. 

\bibliographystyle{unsrt}
\bibliography{nc-qt}

\end{document}